\documentclass{iopjournal}
\usepackage{amsmath,amssymb,bm,booktabs}
\usepackage{hyperref}
\usepackage{mathrsfs}
\newcommand{\ii}{\mathrm i}
\newcommand{\dd}{\mathrm d}
\newcommand{\Tr}{\operatorname{Tr}}
\newcommand{\diag}{\operatorname{diag}}
\newcommand{\Rea}{\operatorname{Re}}
\newcommand{\Ima}{\operatorname{Im}}
\newcommand{\cL}{\mathcal L}
\newcommand{\cD}{\mathcal D}
\newcommand{\cP}{\mathcal P}
\newcommand{\cQ}{\mathcal Q}
\newcommand{\epsb}{\varepsilon_b}
\newcommand{\gm}{\gamma_m}
\newcommand{\gc}{\gamma_c}
\newtheorem{proposition}{Proposition}
\makeatletter
\@ifundefined{proof}{\newenvironment{proof}{\par\noindent\textit{Proof.} }{\hfill$\square$\par}}{}
\makeatother
\allowdisplaybreaks
\begin{document}
\articletype{}
\title{Orbital-coherence memory in electrical transport and its self-consistent Maxwell response}
\author{Pengfei Zhu}

\affil{Bundesanstalt f\"ur Materialforschung und -pr\"ufung (BAM), 12205 Berlin,Germany}

\email{pengfei.zhu@bam.de}

\begin{abstract}
An internal electronic coherence changes electrical transport only when it feeds back on the physical current. We develop a reduced description of this feedback and distinguish its exact consequences from the assumptions required for a microscopic realization. An opposite-parity two-orbital Hamiltonian identifies a current-coupled orbital force. Under an explicitly stated moment closure and collision ansatz, it yields a passive three-state response with an exponentially damped oscillatory memory. The resulting conductivity redistributes a fixed optical weight and predicts a nonmonotonic dependence of dc conduction on coherence damping at fixed remaining parameters. The orbital reduction constrains the coupling and internal frequency through a shared bare splitting, providing a test beyond unconstrained spectral fitting. A charge-conserving spatial extension couples to Maxwell equations through a positive energy balance and separates the constitutive response from longitudinal collective and transverse electromagnetic modes. We formulate complex-conductivity and transmission tests, including the limits of parameter identification over a finite bandwidth. The model supplies a conditional connection between orbital dynamics, transport memory, and electromagnetic observables; conductivity alone does not establish a uniquely quantum mechanism, and a material-specific collision calculation and closure-error assessment remain necessary.
\end{abstract}

\section{Introduction}
The Drude law reduces electrical transport to one relaxing current moment. An additional coherence lifetime does not change that law if the internal state is dynamically decoupled from the current. The physical question is therefore which microscopic operator exerts a coherent force on the current, and how that force changes a measured electromagnetic response.

Linear response, projection-generated memory, and conserving relaxation approximations are established frameworks \cite{kubo1957,mori1965,mermin1970}. We use these tools to examine a particular current--orbital-coherence reduction. The contribution developed here is the connection between an identifiable orbital force, a constrained memory conductivity, and its charge-conserving Maxwell embedding. Neither the use of a density matrix nor the presence of damped complex eigenvalues constitutes a new mechanism by itself.

The analysis separates three statements. Commutators and the physical current are fixed by a specified two-orbital Hamiltonian. A finite moment and collision closure is then imposed, rather than inferred from those commutators. Within that closure, the memory, energy balance, optical-weight identity, and electromagnetic dispersion follow exactly. This separation is essential when using a fitted conductivity to infer electronic coherence.

We use $e^{\ii\bm k\cdot\bm r-\ii\omega t}$ and $s=-\ii\omega$, with Laplace response initially defined for $\Rea s>0$. The carrier charge is $-e$, $e>0$. Fields are classical mean fields and the matter response is linearized about a homogeneous stationary equilibrium. Positive constants $\epsb$ and $\mu_b$ describe background response within the selected frequency window. They exclude the carrier sector treated explicitly below.

\begin{figure}[htbp]
	\centering
	\includegraphics[width=1.0\textwidth]{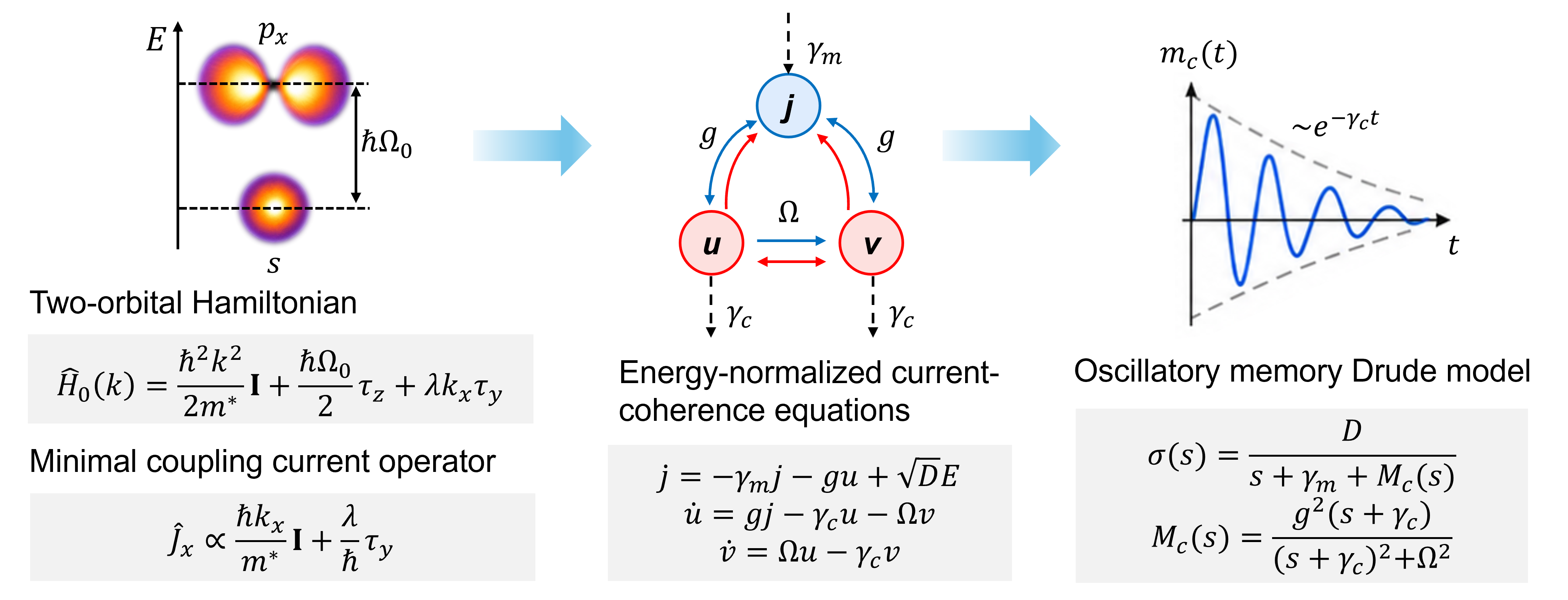}
	\caption{...}
	\label{fig1}
\end{figure}

\section{A microscopic orbital force and the reduced transport model}
\subsection{Hamiltonian, current, and exact moment identities}
A minimal electronic mechanism is a conducting two-orbital system with opposite-parity orbitals, for example an $s$--$p_x$ doublet in one crystal symmetry channel. In an orbital basis with inversion $\mathscr P=\tau_z$ and spinless time reversal $\mathscr T=\mathscr K$, consider
\begin{equation}\label{eq:orbital_H}
 \hat H_0=\sum_{\bm k}\hat c_{\bm k}^{\dagger}
 \left[
 \frac{\hbar^2k^2}{2m^*}\mathbf{I}+
 \frac{\hbar\Omega_0}{2}\tau_z+\lambda k_x\tau_y
 \right]\hat c_{\bm k}.
\end{equation}
The odd-in-$k_x$ hybridization is both inversion and time-reversal invariant because $\tau_y$ is odd under each of $\mathscr P$ and $\mathscr K$. Minimal substitution $k_x\mapsto k_x+eA_x/\hbar$ gives the physical current, at $A_x=0$,
\begin{equation}\label{eq:orbital_current}
 \hat J_x=-\frac{e}{\mathcal V}\sum_{\bm k}\hat c_{\bm k}^{\dagger}
 \left(\frac{\hbar k_x}{m^*}\mathbf{I}+
 \frac{\lambda}{\hbar}\tau_y\right)\hat c_{\bm k}.
\end{equation}
The second term is essential: it is the interorbital contribution to the physical current, not an optional additive Lorentz oscillator.

This is a \emph{doped, partially occupied} carrier-sector construction. A completely filled isolated band sector has no low-frequency Drude weight; interband absorption further requires occupied-to-empty transitions. Such an insulating sector does not establish the proposed transport mechanism without mobile carriers.

Define the homogeneous moments
\begin{equation}\label{eq:orbital_moments}
 K=\frac1{\mathcal V}\sum_{\bm k}\langle\hat c_{\bm k}^{\dagger}k_x\hat c_{\bm k}\rangle,\quad
 X=\frac1{\mathcal V}\sum_{\bm k}\langle\hat c_{\bm k}^{\dagger}\tau_x\hat c_{\bm k}\rangle,\quad
 Y=\frac1{\mathcal V}\sum_{\bm k}\langle\hat c_{\bm k}^{\dagger}\tau_y\hat c_{\bm k}\rangle,\quad
 Q=\frac1{\mathcal V}\sum_{\bm k}\langle\hat c_{\bm k}^{\dagger}k_x\tau_z\hat c_{\bm k}\rangle .
\end{equation}
In this fixed parity-orbital gauge, $X=2\mathcal V^{-1}\sum_{\bm k}\Re\rho_{sp}(\bm k)$ and $Y=-2\mathcal V^{-1}\sum_{\bm k}\Im\rho_{sp}(\bm k)$, up to the sign convention for $\tau_y$. They are therefore the real and imaginary quadratures of an interorbital electronic density-matrix coherence. The reversible equations following from \eqref{eq:orbital_H} are
\begin{equation}\label{eq:orbital_bloch}
 \dot X=\frac{2\lambda}{\hbar}Q-\Omega_0Y,\qquad
 \dot Y=\Omega_0X,\qquad
 \dot J_x=\frac{ne^2}{m^*}E_x-\frac{e\lambda\Omega_0}{\hbar}X .
\end{equation}
The last equality follows from differentiating the \emph{total} current \eqref{eq:orbital_current}; it is precisely the $\tau_y$ current vertex together with $\dot Y=\Omega_0X$ that provides the reciprocal current--coherence feedback.

The identities above use kinetic wave vector in a uniform electric field and vanishing momentum-space boundary terms. They identify a feedback channel but do not close the hierarchy: $Q$ is an independent moment, and its evolution generally generates further moments. The quantities $X,Y$ denote deviations from equilibrium. A fixed parity-orbital basis makes their interpretation explicit; in a general band gauge, coherence must instead be described through the appropriate band projectors.

\begin{figure}[htbp]
	\centering
	\includegraphics[width=0.8\textwidth]{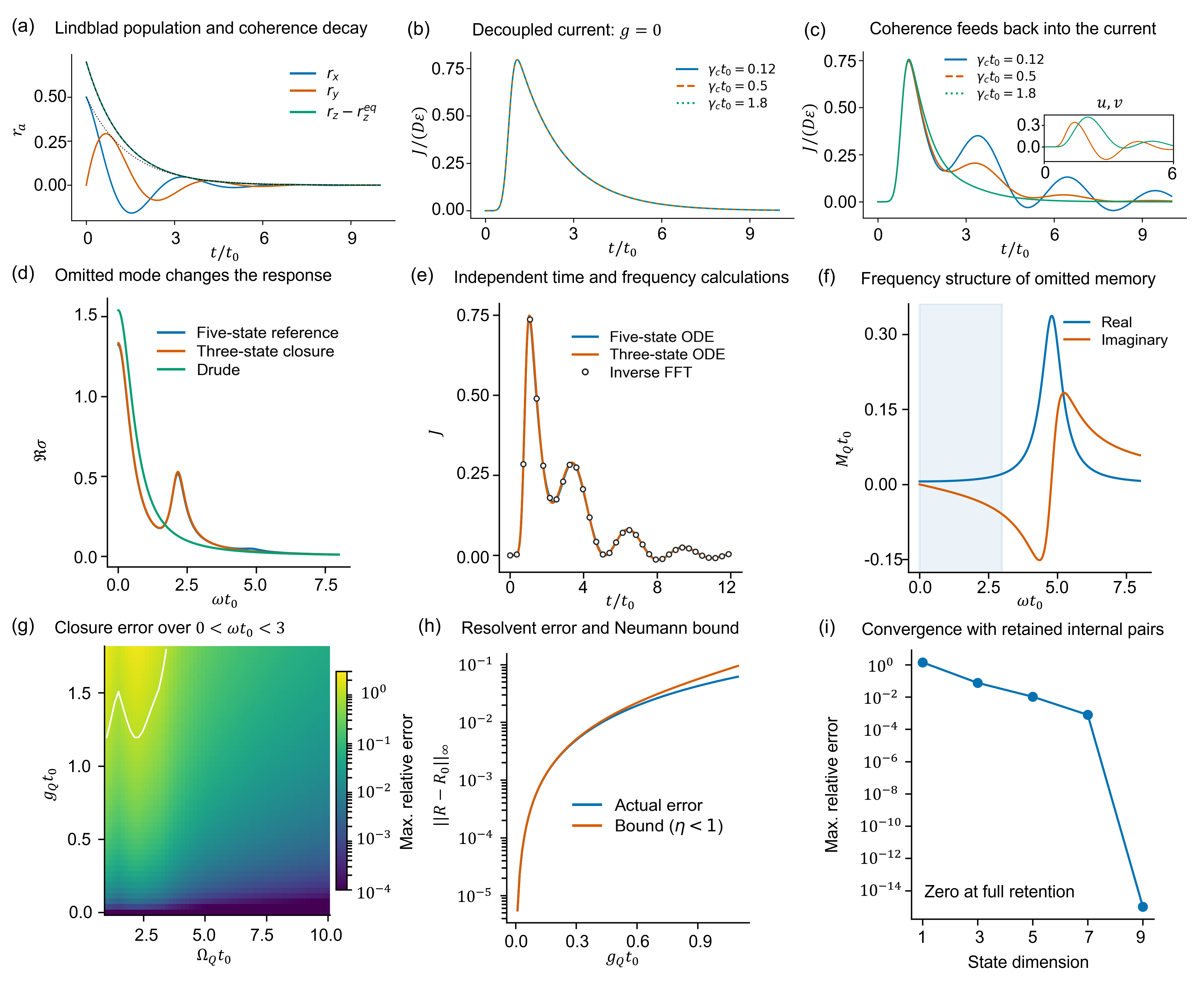}
	\caption{...}
	\label{fig2}
\end{figure}

\subsection{Conditional moment closure and parameter constraint}
Set $Q=\zeta K+Q_\perp$, where $\zeta$ is evaluated for a specified weak-drift perturbation as in Eq.~\eqref{eq:zeta}. A drift derivative is not automatically the equilibrium-metric projection coefficient. Dropping $Q_\perp$ is therefore an additional dynamical approximation, to be tested on the response window. Also assume that the collision operator acts diagonally on the retained \emph{total current} and the two quadratures, as in Eq.~\eqref{eq:orbital_rates}.

For a weak drift perturbation, introduce the response-weighted equilibrium orbital polarization
\begin{equation}\label{eq:zeta}
 \zeta=
 \left.
 \frac{\partial_{\delta k}\langle Q\rangle_{\delta k}}
 {\partial_{\delta k}\langle K\rangle_{\delta k}}
 \right|_{\delta k=0},
 \qquad Q=\zeta K+Q_{\perp}.
\end{equation}
Here $\langle\cdot\rangle_{\delta k}$ denotes the equilibrium distribution shifted infinitesimally in the $x$ direction. We restrict attention to drift states with $\zeta<0$, as expected when the lower $\tau_z=-1$ orbital dominates the participating carrier response; the sign must be evaluated for the specified distribution. Assume, for the moment, that the $Q_{\perp}$ memory is small on the selected response window and that the projected collision operator obeys
\begin{equation}\label{eq:orbital_rates}
 \left.\dot J_x\right|_{\mathrm{coll}}=-\gm J_x,\qquad
 \left.\dot X\right|_{\mathrm{coll}}=-\gc X,\qquad
 \left.\dot Y\right|_{\mathrm{coll}}=-\gc Y .
\end{equation}
Using \eqref{eq:orbital_current} to eliminate $K$, one then obtains
\begin{align}\label{eq:orbital_raw}
 \dot J_x&=-\gm J_x-aX+DE_x,\\
 \dot X&=bJ_x-\gc X-\Omega_{\mathrm{eff}}Y,\nonumber\\
 \dot Y&=\Omega_0X-\gc Y,\nonumber
\end{align}
where
\begin{equation}\label{eq:orbital_map}
 D=\frac{ne^2}{m^*},\qquad
 a=\frac{e\lambda\Omega_0}{\hbar},\qquad
 b=-\frac{2m^*\lambda\zeta}{e\hbar^2},\qquad
 \Omega_{\mathrm{eff}}=\Omega_0+\frac{2m^*\lambda^2\zeta}{\hbar^3}.
\end{equation}
Provided $ab>0$ and $\Omega_0\Omega_{\mathrm{eff}}>0$, a positive algebraic rescaling of $J_x,X,Y$ brings \eqref{eq:orbital_raw} exactly to the normal form \eqref{eq:min_j}--\eqref{eq:min_v}, with
\begin{equation}\label{eq:orbital_reduced_map}
 \begin{aligned}
 g^2&=ab=-\frac{2m^*\lambda^2\Omega_0\zeta}{\hbar^3},
 &\Omega&=\sqrt{\Omega_0\Omega_{\mathrm{eff}}},&
 \gm&=\tau_{\mathrm{tr}}^{-1},\\
 \gc&=\frac{\Gamma_{s\to p}+\Gamma_{p\to s}}{2}
       +\gamma_{\phi}+\gamma_{k\mathrm{-mix}} .
 \end{aligned}
\end{equation}
The last relation is a possible weak-coupling, secular parametrization of interorbital dephasing: population transfer, pure dephasing, and momentum-changing orbital dephasing all contribute. Beyond that regime, $\gm,\gc$ are entries of a frequency-dependent collision matrix rather than constants. Under these closure assumptions, Eq.~\eqref{eq:orbital_reduced_map} gives a conditional parameter map: $u$ is a rescaled real $s$--$p_x$ coherence, $v$ its imaginary quadrature, $\hbar\Omega_0$ is their orbital splitting, and $g$ is fixed by odd-parity hybridization and the equilibrium orbital polarization.

\begin{figure}[htbp]
	\centering
	\includegraphics[width=1.0\textwidth]{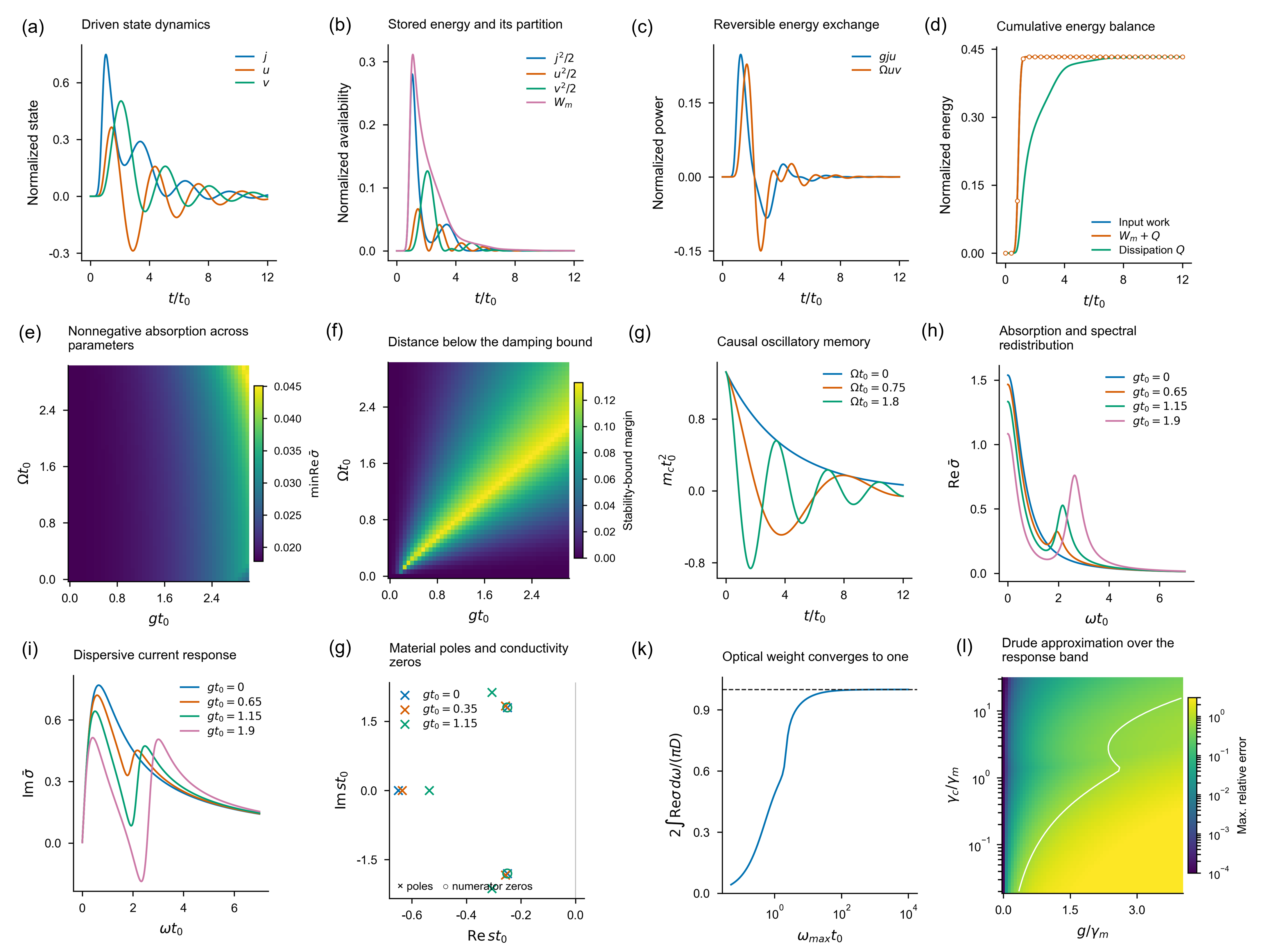}
	\caption{...}
	\label{fig3}
\end{figure}

The algebraic rescaling is explicit. With $j=J_x/\sqrt D$, choose
\begin{equation}\label{eq:scales}
 u=\frac{aX}{g\sqrt D},\qquad
 v=\frac{a\Omega_{\mathrm{eff}}Y}{g\Omega\sqrt D}.
\end{equation}
For nonzero $g,\Omega$ and positive $ab,\Omega_0\Omega_{\mathrm{eff}}$, these coordinates yield the normal form below. This proves the existence of a positive quadratic storage function for the closed equations. Identifying it with microscopic equilibrium availability additionally requires the susceptibility metric; that identification has not been derived from the Hamiltonian alone.

The same map implies a constraint absent from a generic three-pole fit:
\begin{equation}\label{eq:orbital_constraint}
 g^2+\Omega^2=\Omega_0^2,\qquad 0<g^2<\Omega_0^2.
\end{equation}
Here $\Omega_0>0$ and a noninverted, negative $\zeta$ have been assumed. Thus $g$ and $\Omega$ cannot be varied independently at fixed bare splitting within this closure. A fitted parameter set inconsistent with an independently determined $\Omega_0$ falsifies this particular reduction, even if it fits the optical data.

The damping assumption is equally restrictive. Write the physical current as $J_x=-\alpha K-\beta Y$, with $\alpha=e\hbar/m^*$ and $\beta=e\lambda/\hbar$. A bath satisfying $\dot K|_{\rm coll}=-\gm K$ and $\dot Y|_{\rm coll}=-\gc Y$ instead gives
\begin{equation}\label{eq:collision_cross}
 \dot J_x|_{\rm coll}=-\gm J_x+\beta(\gc-\gm)Y.
\end{equation}
The additional term changes the reduced response. A two-level dephasing formula therefore does not establish the assumed diagonal current collision operator. The rates in Eq.~\eqref{eq:orbital_reduced_map} are a conditional parametrization, not a reservoir-derived result for this conductor. Appendix~\ref{app:projection} states how omitted dynamics and vertex corrections enter.

\begin{figure}[htbp]
	\centering
	\includegraphics[width=1.0\textwidth]{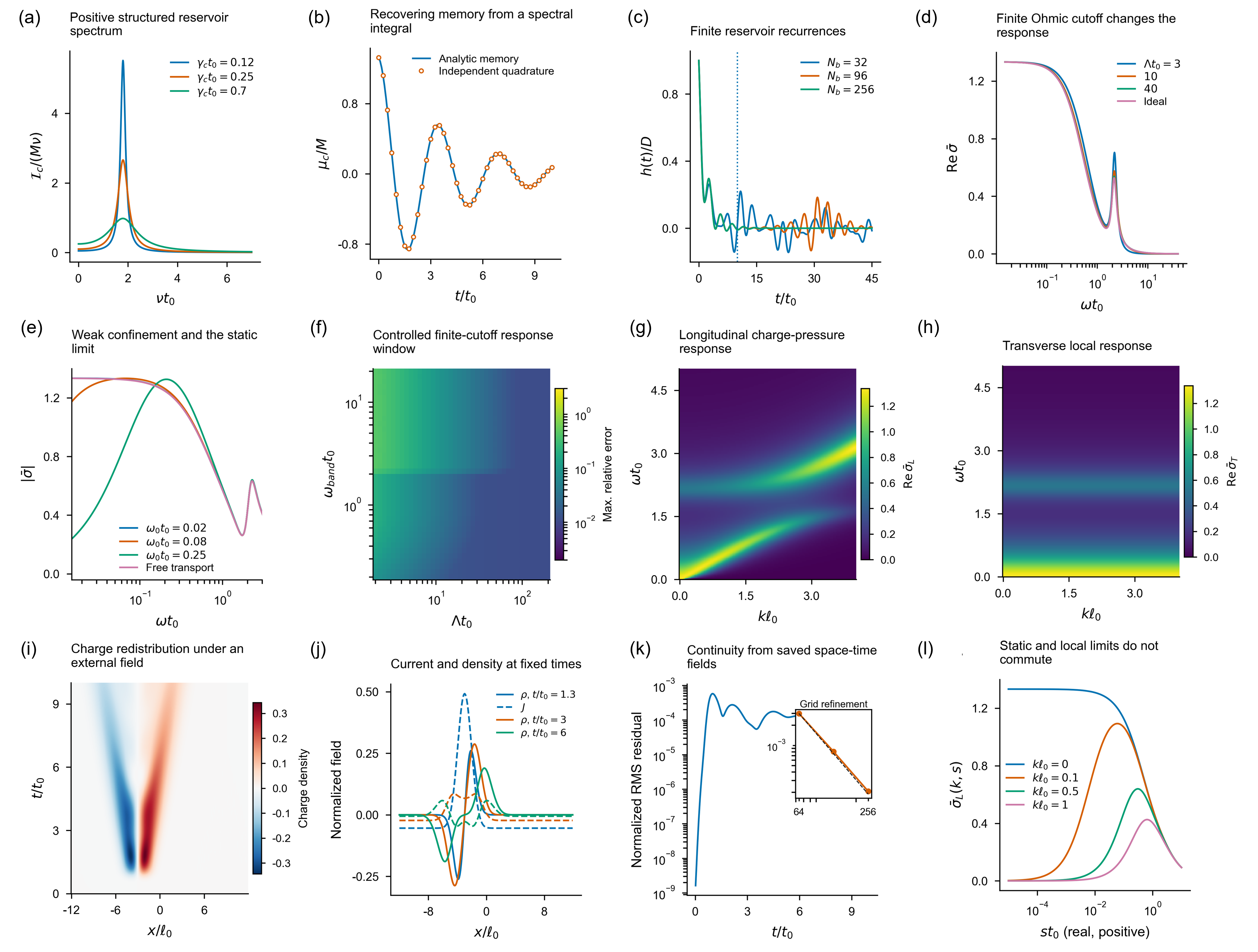}
	\caption{...}
	\label{fig4}
\end{figure}

\subsection{Energy-normalized current--coherence equations}
For a symmetry-resolved scalar channel, let $D>0$, $g,\Omega\geq0$, and $\gm,\gc>0$. The retained current is $J=\sqrt D j$. The reduced electric port acts on this total-current coordinate; the coherence contribution already contained in $J$ must not be added a second time.
The equations are
\begin{align}
 \dot j&=-\gm j-gu+\sqrt D E,\label{eq:min_j}\\
 \dot u&=g j-\gc u-\Omega v,\label{eq:min_u}\\
 \dot v&=\Omega u-\gc v,\qquad J=\sqrt D j.\label{eq:min_v}
\end{align}
Equivalently,
\begin{equation}\label{eq:A}
 A=\begin{pmatrix}-\gm&-g&0\\g&-\gc&-\Omega\\0&\Omega&-\gc\end{pmatrix},\quad
 b=\sqrt D\begin{pmatrix}1\\0\\0\end{pmatrix},\quad
 \Gamma=\diag(\gm,\gc,\gc).
\end{equation}
The opposite signs of the $g$ terms are fixed by reversible energy exchange. With equal signs the coupling would contribute an indefinite term to $\dot W_m$ and would require a different metric or a different physical interpretation. Here
\begin{equation}\label{eq:energy_min}
 W_m=\frac12(j^2+u^2+v^2),\qquad
 \dot W_m=EJ-\gm j^2-\gc(u^2+v^2).
\end{equation}
For a purely parabolic transport current, $J^2/(2D)=mn\bar v^2/2$. For the orbital total current, Eq.~\eqref{eq:energy_min} is a reduced storage identity; its identification with microscopic availability requires a separate susceptibility calculation.

The equations define an effective model even when the orbital map is not valid. In that case $g$ and $\Omega$ are independent effective coefficients, and an electronic-coherence interpretation requires a separate microscopic identification. A general projected bath may also introduce cross damping, unequal quadrature decay, or additional electric vertices. Those terms cannot be removed by naming the retained variables ``coherence''.

\begin{figure}[htbp]
	\centering
	\includegraphics[width=1.0\textwidth]{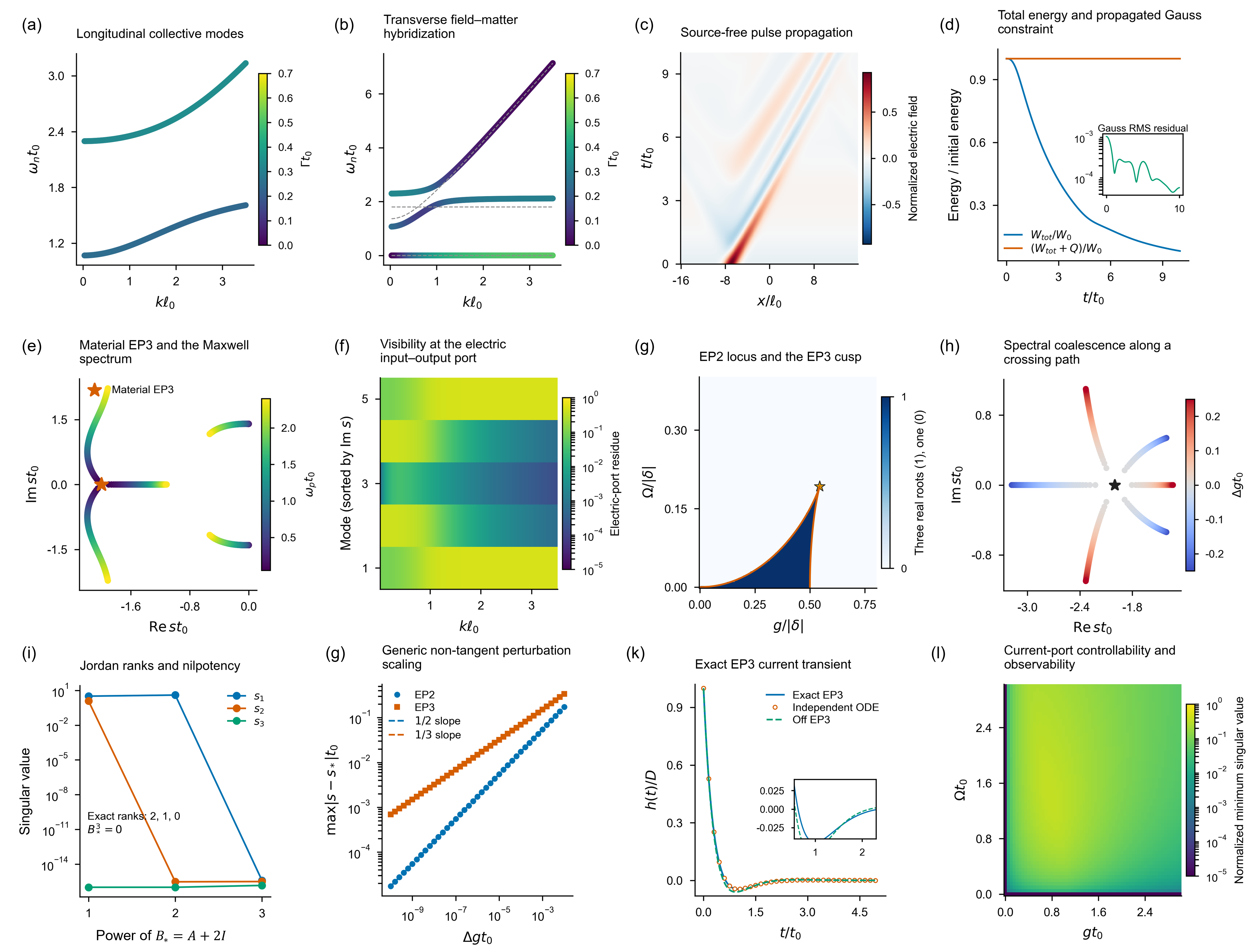}
	\caption{...}
	\label{fig5}
\end{figure}

\section{Memory conductivity and testable transport consequences}
\subsection{Exact reduced conductivity}

Solving the last two equations gives
\begin{equation}\label{eq:uv}
 u=\frac{g(s+\gc)}{N(s)}j,\qquad
 v=\frac{g\Omega}{N(s)}j,\qquad N(s)=(s+\gc)^2+\Omega^2.
\end{equation}
Substitution into the current equation yields
\begin{align}
 M_c(s)&=\frac{g^2(s+\gc)}{(s+\gc)^2+\Omega^2},\label{eq:M}\\
 \sigma(s)&=\frac{D}{s+\gm+M_c(s)}
 =\frac{D N(s)}{P(s)},\label{eq:sigma}\\
 P(s)&=(s+\gm)N(s)+g^2(s+\gc)=\det(sI-A).\label{eq:P}
\end{align}
In the frequency domain,
\begin{equation}\label{eq:sigma_omega}
 \sigma(\omega)=\frac{D[(\gc-\ii\omega)^2+\Omega^2]}
 {(\gm-\ii\omega)[(\gc-\ii\omega)^2+\Omega^2]+g^2(\gc-\ii\omega)}.
\end{equation}
The elimination modifies the inverse current response through a memory function. This dynamical interpretation does not imply that a measured rational response has a unique realization.

\subsection{The causal memory kernel}

For $u(0)=v(0)=0$, inverse transformation of \eqref{eq:M} gives
\begin{equation}\label{eq:memory_time}
 \dot j(t)+\gm j(t)+\int_0^t m_c(t-t')j(t')\,\dd t'=\sqrt D E(t),
 \qquad m_c(t)=g^2e^{-\gc t}\cos(\Omega t).
\end{equation}
The kernel has units $\mathrm{s}^{-2}$ and is causal. It changes sign when $\Omega\neq0$, which does not violate passivity; positivity constrains the full frequency-domain work, not the pointwise sign of a memory kernel. For a nonzero initial coherence, an additional term appears on the right:
\begin{equation}
 -g e^{-\gc t}[u(0)\cos(\Omega t)-v(0)\sin(\Omega t)].
\end{equation}
An experiment preparing internal coherence before applying a probe therefore needs an initial-state response in addition to a stationary conductivity.

\subsection{Absorption, DC response, and spectral weight}

The memory can be written
\begin{equation}\label{eq:Mpartial}
 M_c(-\ii\omega)=\frac{g^2}{2}\left[
 \frac1{\gc-\ii(\omega-\Omega)}+
 \frac1{\gc-\ii(\omega+\Omega)}\right].
\end{equation}
Its real part is
\begin{equation}\label{eq:ReM}
 \Rea M_c(-\ii\omega)=\frac{g^2\gc}{2}
 \left[\frac1{\gc^2+(\omega-\Omega)^2}
 +\frac1{\gc^2+(\omega+\Omega)^2}\right]\geq0.
\end{equation}
Consequently,
\begin{equation}\label{eq:absorption}
 \Rea\sigma(\omega)=D\frac{\gm+\Rea M_c(-\ii\omega)}
 {[\gm+\Rea M_c(-\ii\omega)]^2+[-\omega+\Ima M_c(-\ii\omega)]^2}\geq0.
\end{equation}
The DC response and the large-$s$ expansion are
\begin{align}
 \sigma_{\mathrm{DC}}&=\frac{D}{\gm+g^2\gc/(\gc^2+\Omega^2)},\label{eq:dc}\\
 \sigma(s)&=\frac{D}{s}-\frac{D\gm}{s^2}
 +\frac{D(\gm^2-g^2)}{s^3}+O(s^{-4}).\label{eq:high_s}
\end{align}
The conservative coupling reduces dc conduction at fixed positive rates while preserving the retained optical weight. It redistributes that weight over frequency.

\begin{proposition}[Spectral-weight identity]\label{prop:sumrule}
For positive damping and the minimal model,
\begin{equation}\label{eq:sumrule}
 \int_0^\infty\Rea\sigma(\omega)\,\dd\omega=\frac{\pi D}{2}.
\end{equation}
\end{proposition}
\textit{Proof.} The impulse response is $h(t)=b^Te^{At}b$ for $t\geq0$, with $h(0^+)=D$ and exponential decay. The cosine-transform inversion at $t=0^+$ gives \eqref{eq:sumrule}. Equivalently it follows by a contour integral using analyticity, reality, and the leading term $\sigma(\omega)\sim\ii D/\omega$.

This identity is exact for the rational model integrated over all frequencies. Its use as a material sum rule requires the model to represent the relevant spectral weight over that range. A low-energy fit cannot be extrapolated to arbitrarily high frequency simply to enforce a bare-electron sum rule.

\subsection{Zeros and cancellation}

When $g\Omega\neq0$, the zeros of $N$ are $s_z=-\gc\pm\ii\Omega$, and $P(s_z)=g^2(s_z+\gc)\neq0$. They are conductivity zeros, not poles. If $g=0$, $N$ cancels and only the Drude pole remains. If $\Omega=0$,
\begin{equation}\label{eq:Omega0}
 \sigma(s)=\frac{D(s+\gc)}{(s+\gm)(s+\gc)+g^2},
\end{equation}
and one decoupled coherence coordinate is unobservable. Multiplying by a denominator without checking these cancellations can create fictitious modes in a response calculation.

\subsection{Uniform convergence on a response window}

Let $\sigma_D(s)=D/(s+\gm)$ and define $\eta_D(s)=M_c(s)/(s+\gm)$. Then
\begin{equation}\label{eq:Drude_error}
 \frac{\sigma(s)-\sigma_D(s)}{\sigma_D(s)}
 =-\frac{\eta_D(s)}{1+\eta_D(s)}.
\end{equation}
On a compact set where $\sup|\eta_D|=\eta<1$,
\begin{equation}\label{eq:Drude_bound}
 \sup\left|\frac{\sigma-\sigma_D}{\sigma_D}\right|
 \leq\frac{\eta}{1-\eta}.
\end{equation}
Thus $g\to0$ at fixed $\gm,\gc,\Omega$ yields the local Drude law uniformly away from singular limits. A fast-coherence limit also gives that law if $M_c\to0$. However, if $g^2/\gc\to\Lambda>0$ while $\Omega/\gc\to0$, then $M_c(s)\to\Lambda$ at fixed $s$ and the limiting rate is $\gm+\Lambda$. Eliminating an internal variable can renormalize transport rather than remove its effect.

A Drude limit is therefore a small-memory limit on a specified frequency window. It is not synonymous with $\hbar\to0$: even a degenerate parabolic Fermi gas has the Drude weight $ne^2/m$ when its first current moment relaxes with one rate.

\subsection{Dephasing dependence and the meaning of quantum coherence}

Holding $D,\gm,g,\Omega$ fixed,
\begin{equation}\label{eq:dc_derivative}
 \frac{\partial\sigma_{\mathrm{DC}}}{\partial\gc}
 =-\frac{Dg^2(\Omega^2-\gc^2)}
 {(\gc^2+\Omega^2)^2[\gm+g^2\gc/(\gc^2+\Omega^2)]^2}.
\end{equation}
The memory contribution to the dc rate is largest at $\gc=\Omega$. Increasing dephasing suppresses conduction for $\gc<\Omega$ but increases it for $\gc>\Omega$. This turnover is a definite model prediction under the fixed-parameter condition. A temperature scan that also changes carrier density, momentum scattering, or $g$ does not isolate this derivative.

This turnover is a conditional response prediction, not a guaranteed temperature trend. Temperature, gating, and disorder can change $D,\gm,g,\Omega$ simultaneously. In the orbital construction, Eq.~\eqref{eq:orbital_constraint} further restricts admissible parameter changes.

All deterministic equations above also admit a classical passive realization. Accordingly, a non-Drude line shape does not establish an electronic quantum origin. Such an interpretation requires independent orbital energies and occupations, a calculated collision operator, and a controlled projection, or a separate coherence-sensitive measurement. A structured oscillator bath provides another realization of the same memory (Appendix~\ref{app:bath}); it does not identify the auxiliary pair as electronic density-matrix elements.

\section{Charge conservation and self-consistent electromagnetic response}
\subsection{Conserving spatial extension}
Let $\rho$ now denote the macroscopic perturbation of charge density, distinguished from the microscopic density operator by the absence of a hat. Write the charge availability as $\rho^2/(2C_\rho)$ with $C_\rho>0$. For a simple electron gas, $C_\rho=e^2\partial n/\partial\mu$. Define
\begin{equation}\label{eq:ch}
 y=\frac{\rho}{\sqrt{C_\rho}},\qquad
 c_h^2=\frac{D}{C_\rho},\qquad \bm J=\sqrt D\bm j.
\end{equation}
The coefficient $c_h$ is a compressibility-controlled hydrodynamic speed. Its relation to Fermi velocity depends on the kinetic regime and closure; the collisionless high-frequency pressure coefficient is not generally the isothermal compressibility coefficient.

The conserving equations are
\begin{align}
 \partial_t y&=-c_h\nabla\cdot\bm j,\label{eq:nonlocal_y}\\
 \partial_t\bm j&=-\gm\bm j-g\bm u-c_h\nabla y+\sqrt D\bm E,\label{eq:nonlocal_j}\\
 \partial_t\bm u&=g\bm j-\gc\bm u-\Omega\bm v,\label{eq:nonlocal_u}\\
 \partial_t\bm v&=\Omega\bm u-\gc\bm v.\label{eq:nonlocal_v}
\end{align}
Multiplying the first equation by $\sqrt{C_\rho}$ recovers exact charge continuity. The pressure force in the current equation is $-c_h\nabla y=-\sqrt D\nabla(\rho/C_\rho)$, which is the force from the charge availability after normalization.

For real fields, the local matter energy identity is
\begin{equation}\label{eq:nonlocal_energy}
 \partial_t\left[\frac{y^2+|\bm j|^2+|\bm u|^2+|\bm v|^2}{2}\right]
 +\nabla\cdot(c_h y\bm j)
 =\bm E\cdot\bm J-\gm|\bm j|^2-\gc(|\bm u|^2+|\bm v|^2).
\end{equation}
The pressure terms therefore exchange local density and kinetic availability while producing a matter energy flux. A damping term $-\gamma_ny$ would destroy charge conservation unless an explicit charge-exchange reservoir were introduced.

\subsection{Longitudinal and transverse conductivity}

For $k>0$, define $P_L=\hat{\bm k}\hat{\bm k}^{T}$ and $P_T=I-P_L$. Isotropy gives
\begin{equation}\label{eq:sigma_tensor}
 \bm\sigma(\bm k,s)=\sigma_L(k,s)P_L+\sigma_T(s)P_T.
\end{equation}
The transverse response is \eqref{eq:sigma}; the longitudinal response follows from $y=-\ii c_h k j_L/s$:
\begin{align}
 \sigma_L(k,s)&=\frac{D}{s+\gm+M_c(s)+c_h^2k^2/s}\label{eq:sigmaL}\\
 &=\frac{D s N(s)}{Q_L(k,s)},\qquad
 Q_L=(s^2+\gm s+c_h^2k^2)N(s)+g^2s(s+\gc).\label{eq:QL}
\end{align}
The electric field here is the total field acting on the carriers. Coulomb self-consistency is not yet inserted into $Q_L$. This separation prevents double counting of screening.

This pressure closure is an additional long-wavelength approximation, not a consequence of the uniform two-orbital commutators. The scalar coefficients below assume an isotropic extension, such as equivalent orbital channels along the three axes. A single $s$--$p_x$ doublet instead describes one anisotropic channel and requires a tensor response outside that channel.

\subsection{Maxwell coupling and total energy}
The coupled equations are
\begin{align}
 \partial_t\bm B&=-\nabla\times\bm E,\qquad \bm B=\mu_b\bm H,\\
 \epsb\partial_t\bm E&=\nabla\times\bm H-\bm J,\label{eq:Maxwell}\\
 \nabla\cdot(\epsb\bm E)&=\rho,\qquad \nabla\cdot\bm B=0.\label{eq:Gauss}
\end{align}
The background displacement $\epsb\bm E$ excludes the retained carrier current. If that current were also included in a fitted background permittivity, the same response would be counted twice. Bound responses outside the model window can be included in the background only with consistent dispersion and loss treatment. Constant positive $\epsb$ is a finite-window idealization.

Taking the divergence of Amp\`ere's equation and using continuity gives
\begin{equation}\label{eq:constraint_prop}
 \partial_t[\epsb\nabla\cdot\bm E-\rho]=0,\qquad
 \partial_t(\nabla\cdot\bm B)=0.
\end{equation}
Hence compatible initial data remain compatible. A numerical implementation should enforce these constraints rather than treat spurious longitudinal electromagnetic variables as independent physical modes.

\subsection{Total energy balance}

Combining Poynting's theorem with \eqref{eq:nonlocal_energy} gives
\begin{align}
 \partial_t W_{\mathrm{tot}}+\nabla\cdot\left(\bm E\times\bm H+c_hy\bm j\right)
 &=-\gm|\bm j|^2-\gc(|\bm u|^2+|\bm v|^2),\label{eq:total_energy}\\
 W_{\mathrm{tot}}&=\frac{\epsb|\bm E|^2+\mu_b|\bm H|^2
 +y^2+|\bm j|^2+|\bm u|^2+|\bm v|^2}{2}.
\end{align}
All terms have units of energy density.

For periodic boundaries, or boundaries without incoming total energy flux, the integrated energy is nonincreasing. Charge continuity propagates Gauss's constraint. A reflecting carrier boundary requires $\bm j\cdot\bm n=0$; the internal pair has no spatial derivatives and needs no additional spatial boundary condition in this closure. These statements establish consistency of the effective Maxwell--matter evolution. They do not assert a microscopic relativistic propagation theory.

\subsection{Longitudinal collective and transverse electromagnetic modes}
For $k>0$, Gauss's law fixes $E_L=\rho/(\ii k\epsb)$. Define
\begin{equation}\label{eq:plasma}
 \omega_p^2=\frac{D}{\epsb},\qquad
 v_b=\frac1{\sqrt{\mu_b\epsb}}.
\end{equation}
Substitution into the density--current equations gives the longitudinal characteristic polynomial
\begin{equation}\label{eq:PL}
 \mathcal P_L(k,s)=
 (s^2+\gm s+c_h^2k^2+\omega_p^2)N(s)
 +g^2s(s+\gc).
\end{equation}
For $g=0$, the physical carrier plasmon satisfies
\begin{equation}\label{eq:plasmon_g0}
 s^2+\gm s+c_h^2k^2+\omega_p^2=0,
\end{equation}
while the decoupled internal pair has its own dark modes. In frequency language the uncoupled relation is $\omega(\omega+\ii\gm)=\omega_p^2+c_h^2k^2$.

The same polynomial follows from the longitudinal effective permittivity
\begin{equation}\label{eq:epsilonL}
 \varepsilon_L^{\mathrm{eff}}(k,s)=\epsb+\frac{\sigma_L(k,s)}s,
\end{equation}
provided zeros, poles, and removable factors are checked. At fixed $k>0$, \eqref{eq:PL} is the characteristic polynomial of the constrained four-state system $(y,j_L,u_L,v_L)$, not an unconstrained six-component electromagnetic determinant.
Each transverse polarization has the characteristic polynomial
\begin{equation}\label{eq:PT}
 \mathcal P_T(k,s)=(s^2+v_b^2k^2)P(s)+\omega_p^2sN(s).
\end{equation}
It has degree five and represents the augmented state $(E_T,H_T,j_T,u_T,v_T)$ at $k>0$. The corresponding reduced dispersion is
\begin{equation}
 s^2+v_b^2k^2+\frac{s\sigma_T(s)}{\epsb}=0.
\end{equation}
At $k=0$, a static magnetic component separates and gives the explicit factor $s$ in \eqref{eq:PT}; the remaining electric--matter subsystem has degree four. If $g=0$, a pair of material poles cancels in the optical response but remains as decoupled internal eigenvalues of the full augmented generator. This distinction between internal modes and optically coupled modes is necessary when comparing a polynomial calculation with a measured spectrum.

The two polynomials describe freely evolving self-consistent fields. Their roots differ from the poles of $\sigma(s)$ measured with an imposed field. This distinction is needed to compare a constitutive fit with plasmon or polariton measurements. The detailed state-matrix checks are given in Appendix~\ref{app:det}.

\begin{figure}[htbp]
	\centering
	\includegraphics[width=1.0\textwidth]{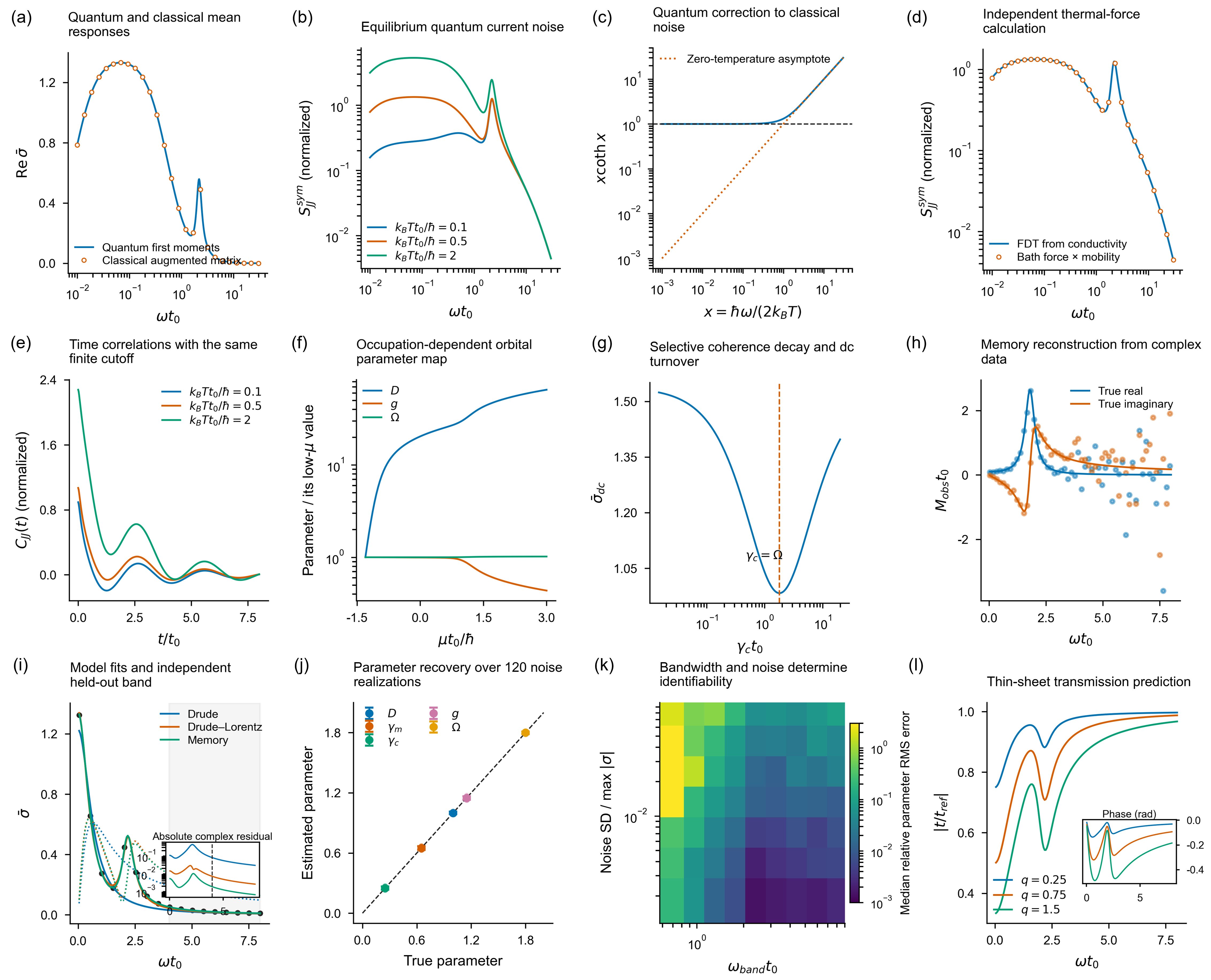}
	\caption{...}
	\label{fig6}
\end{figure}

\section{Optical reconstruction and model discrimination}
\subsection{What complex conductivity can identify}

If $D$ and $\gm$ are known independently or determined from a valid broad spectral range, the experimentally reconstructed quantity
\begin{equation}\label{eq:reconstruct}
 M_{\mathrm{obs}}(s)=\frac{D}{\sigma(s)}-s-\gm
\end{equation}
is zero in an ideal one-pole Drude model and equals \eqref{eq:M} in the proposed closure. Its conjugate pole pair at $-\gc\pm\ii\Omega$ has residues $g^2/2$. This representation is useful for distinguishing a transport memory from a simple additive optical resonance, but inversion amplifies uncertainty near small $\sigma$.

For noiseless knowledge of the rational transfer function with $g>0$, $\Omega>0$, the asymptotic moments identify $D$ and $\gm$ from \eqref{eq:high_s}, and the poles and residues of $M_{\mathrm{obs}}$ identify $\gc,\Omega,g^2$. The sign of $g$ is a coordinate convention and cannot be inferred. Finite bandwidth and noise can make these formally identifiable parameters practically nonidentifiable.

This formal uniqueness does not imply accurate estimation in a restricted THz band. If the memory varies little across the observed band, it can be absorbed into an apparent transport rate and inertia; $g,\gc,\Omega$ then become strongly correlated. Fits should use the complex response with the measurement uncertainty, report parameter correlations or profile intervals, and test robustness to the fitted bandwidth. Inverse-response reconstruction is a diagnostic; fitting the measured complex transmission directly avoids amplifying conductivity-inversion errors near its zeros.

\subsection{Transmission and independent microscopic tests}

For an optically thin homogeneous sheet with sheet conductance $G_s(\omega)=d\sigma_T(\omega)$ at an interface between nonmagnetic media of indices $n_1,n_2$, neglecting multiple internal reflections,
\begin{equation}\label{eq:sheet_t}
 t(\omega)=\frac{2n_1}{n_1+n_2+Z_0G_s(\omega)},\qquad
 \frac{t}{t_0}=\frac{n_1+n_2}{n_1+n_2+Z_0G_s(\omega)},
\end{equation}
where $t_0=2n_1/(n_1+n_2)$ and $Z_0$ is the vacuum impedance. This formula follows from tangential field continuity and the sheet-current jump condition. It applies only when thickness and nonlocal boundary effects can be neglected. A thick or spatially dispersive sample needs the full boundary-value problem and an explicitly specified carrier boundary condition.

Equivalently, for a known thickness and $r=t/t_0$,
\begin{equation}
 \sigma_T(\omega)=\frac{n_1+n_2}{Z_0d}\left[r(\omega)^{-1}-1\right].
\end{equation}
This inversion presumes the stated thin-sheet reference geometry, calibrated phase, and negligible unresolved multiple reflections. A thicker film or substrate etalon requires a transfer-matrix forward model. Far-field transverse transmission does not independently determine the longitudinal compressibility $C_\rho$; a charge-sensitive or finite-wave-vector probe is needed for that part of the model.

The physically discriminating comparison uses a one-pole Drude response, an additive Drude--Lorentz response, and the present memory response under the same measurement model. The inverse-response form is not, by itself, a uniqueness theorem distinguishing all classical or additive realizations. Support for the orbital mechanism requires that independently determined $m^*,\lambda,\Omega_0$, orbital occupation response, and collisions predict both the spectrum and its change under a controlled perturbation. Equation~\eqref{eq:orbital_constraint} supplies an additional check beyond fit quality.

\section{Conclusions}
The central result is a conditional link from a physical orbital current vertex to an oscillatory transport memory and its electromagnetic observables. The two-orbital Hamiltonian identifies the current force exactly; the finite moment and diagonal collision approximations are separate hypotheses. Within those hypotheses, the orbital parameters obey $g^2+\Omega^2=\Omega_0^2$, and the reduced response preserves a positive energy balance while redistributing optical weight. Selective changes of coherence damping produce a calculable dc turnover at fixed remaining parameters.

The conserving spatial extension separates total-field conductivity from self-consistent collective modes and provides a direct route to optical forward modelling. The principal remaining task is material-specific: calculate the orbital response and full collision matrix, quantify the omitted memory, and compare with complex spectra and an independent coherence probe. Exceptional degeneracies of the reduced matrix are optional spectral properties, recorded in Appendix~\ref{app:EP}; neither they nor non-Hermitian terminology are required for the transport mechanism developed here.

\appendix
\section{Exact response, physical vertices, and local continuity}\label{app:response}
A continuum starting point is
\begin{equation}\label{eq:H}
 \hat H[\bm A,\phi]=\sum_{a=1}^{N}
 \left\{\frac{[\hat{\bm p}_a+e\bm A(\hat{\bm r}_a,t)]^2}{2m}
 +V_{\mathrm{lat}}(\hat{\bm r}_a)-e\phi(\hat{\bm r}_a,t)\right\}
 +\hat H_{\mathrm{int}}+\hat H_{\mathrm{internal}}.
\end{equation}
Here $\hat H_{\mathrm{internal}}$ may describe orbital, spin, interband, or other neutral internal degrees of freedom. Its choice is part of the model, not determined by Maxwell equations. If the electromagnetic field is self-consistent, the longitudinal Hartree interaction must be counted either in the matter Hamiltonian or in the self-consistent field, with a specified subtraction; including it twice gives an incorrect collective spectrum.

The charge and matter current are defined by
\begin{equation}\label{eq:charge_current}
 \hat\rho_{\mathrm{ch}}(\bm r)=-e\sum_a\delta(\bm r-\hat{\bm r}_a),\qquad
 \hat{\bm j}(\bm r)=-\frac{\delta\hat H}{\delta\bm A(\bm r)}.
\end{equation}
For the kinetic term the current is the symmetrized velocity density. In particular, for a homogeneous field,
\begin{equation}\label{eq:diamagnetic}
 \hat{\bm J}[\bm A]=\hat{\bm J}^{p}-\frac{Ne^2}{m\mathcal V}\bm A,
 \qquad \hat{\bm J}^{p}=-\frac{e}{m\mathcal V}\sum_a\hat{\bm p}_a,
\end{equation}
where $\mathcal V$ is the sample volume. The second term is the diamagnetic contact contribution. It is essential in a velocity-gauge conductivity calculation. In a crystal, the corresponding contact tensor is the field derivative of the full Hamiltonian; it is not generally $ne^2/m$ within a truncated band description.
The matter state satisfies
\begin{equation}\label{eq:master}
 \partial_t\hat\rho=\cL[\bm A,\phi]\hat\rho
 =-\frac{\ii}{\hbar}[\hat H[\bm A,\phi],\hat\rho]
 +\cD[\bm A,\phi]\hat\rho.
\end{equation}
In a finite-dimensional regularization, a Markovian completely positive, trace-preserving generator has the form
\begin{equation}\label{eq:GKSL}
 \cD\rho=\sum_\alpha\left(
 L_\alpha\rho L_\alpha^\dagger
 -\frac12\{L_\alpha^\dagger L_\alpha,\rho\}\right),
\end{equation}
with rates absorbed into $L_\alpha$ \cite{lindblad1976,gks1976}. For continuum electrons, unbounded operators require additional domain and existence hypotheses; finite-dimensional GKSL results cannot be transferred to that setting without qualification. Weak-coupling Markov limits also require hypotheses on reservoirs and time scales \cite{davies1974}.

Trace preservation guarantees total probability conservation. It does not guarantee local charge conservation, a Gibbs stationary state, detailed balance, or electromagnetic gauge covariance. We impose these separately when they are needed. A thermal stationary state is denoted $\rho_0$; in a closed matter model it is $Z^{-1}\exp[-\beta(\hat H_0-\mu\hat N)]$, with $\beta=(k_BT)^{-1}$. A reservoir approximation must actually leave the selected $\rho_0$ stationary.
\begin{equation}\label{eq:gauge}
 \bm A'=\bm A+\nabla\chi,\quad \phi'=\phi-\partial_t\chi,\quad
 U_\chi=\exp\left[-\frac{\ii e}{\hbar}\sum_a\chi(\hat{\bm r}_a,t)\right],\quad
 \rho'=U_\chi\rho U_\chi^\dagger.
\end{equation}
The Hamiltonian transforms with the time-dependent unitary term, and a covariant dissipator obeys
\begin{equation}\label{eq:Dbathgauge}
 \cD'[U_\chi\rho U_\chi^\dagger]
 =U_\chi\cD[\rho]U_\chi^\dagger.
\end{equation}
A relaxation operator diagonal in an arbitrarily fixed truncated basis need not satisfy this identity. Gauge issues in truncated matter--field models have been examined explicitly in the literature \cite{gustin2023,gustinerratum2025,ventura2017}.

Linearization of the complete generator gives
\begin{equation}\label{eq:linear_generator}
 \partial_t\delta\rho=\cL_0\delta\rho+\delta\cL\,\rho_0,
 \qquad \delta\cL\,\rho_0
 =-\frac{\ii}{\hbar}[\delta H,\rho_0]+\delta\cD\,\rho_0.
\end{equation}
Thus the source has a Hamiltonian vertex and, whenever the bath is field dependent, a dissipative vertex. Dropping the latter while changing gauge changes the model rather than merely changing representation. The same warning applies to a field-dependent measured current.

\subsection{Local continuity and a conserving collision approximation}
For a locally charge-conserving Hamiltonian,
\begin{equation}\label{eq:operator_cont}
 \frac{\ii}{\hbar}[H,\hat\rho_{\mathrm{ch}}]+\nabla\cdot\hat{\bm j}_H=0.
\end{equation}
If the bath produces no charge transport, one must have $\cD^*\hat\rho_{\mathrm{ch}}=0$. More generally, a bath-mediated current is required such that
\begin{equation}\label{eq:bath_cont}
 \cD^*\hat\rho_{\mathrm{ch}}+\nabla\cdot\hat{\bm j}_{\mathrm{bath}}=0.
\end{equation}
Here $\cD^*$ is the Heisenberg adjoint. The physical current in Maxwell equations is the total current appearing in continuity.

Taking an expectation yields the linear Ward constraint
\begin{equation}\label{eq:ward}
 s\,\delta\rho_{\mathrm{ch}}(\bm k,s)+\ii\bm k\cdot\delta\bm J(\bm k,s)=0.
\end{equation}
For a zero-initial-data perturbation this relation must hold for every external potential, including a pure gauge. Global number conservation only gives the $\bm k=0$ part and cannot replace the finite-$\bm k$ identity.

A useful kinetic illustration is a number-conserving BGK collision term
\begin{equation}\label{eq:BGK}
 C[f]=-\gm(f-f_{\mathrm{le}}),\qquad
 \int(f-f_{\mathrm{le}})\,\dd\bm p=0,
\end{equation}
where the local equilibrium contains a chemical-potential perturbation fixed by the local density. Relaxation to a globally fixed distribution would instead spuriously remove density. This is the central conservation issue addressed by Mermin's construction \cite{mermin1970}. Equation \eqref{eq:BGK} is a kinetic approximation and is not asserted to be a GKSL generator on an arbitrary many-electron Hilbert space.
\subsection{The Liouvillian resolvent}
Let $\lambda_a$ denote independent external potentials or other probes, and define
\begin{equation}\label{eq:vertices}
 S_a=\left.\frac{\partial\cL}{\partial\lambda_a}\right|_0\rho_0,
 \qquad C_{ia}=\Tr\left(\rho_0\left.\frac{\partial\hat J_i}{\partial\lambda_a}\right|_0\right).
\end{equation}
The trace-zero subspace contains the perturbation. Assume $\cL_0$ generates a bounded semigroup there, or use a finite regularization in which the following expressions are ordinary matrix resolvents.

For zero initial perturbation, Duhamel's formula and the Laplace transform give
\begin{align}
 \delta\rho(t)&=\int_0^t e^{(t-t')\cL_0}S_a\lambda_a(t')\,\dd t',\\
 \delta\rho(s)&=(s-\cL_0)^{-1}S_a\lambda_a(s),\\
 \delta J_i(s)&=K_{ia}(s)\lambda_a(s),\qquad
 K_{ia}(s)=C_{ia}+\Tr\{\hat J_i(s-\cL_0)^{-1}S_a\}.\label{eq:resolvent}
\end{align}
Repeated indices are summed. In a homogeneous vector-potential gauge, $\bm E=-s\bm A$ and
\begin{equation}\label{eq:sigma_exact}
 \sigma_{ij}(s)=-\frac{1}{s}\left[
 C_{ij}^{A}+\Tr\{\hat J_i(s-\cL_0)^{-1}S_j^A\}\right].
\end{equation}
At finite $\bm k$, the scalar potential and longitudinal vector potential must be combined subject to \eqref{eq:ward}. Simply dividing one arbitrary potential-response matrix by $s$ does not define a gauge-independent longitudinal conductivity.

Equation \eqref{eq:sigma_exact} is the exact response of the specified Markovian model. It is not exact for an unspecified physical reservoir. It also makes clear why there is no unique function $\sigma(\bm k,\omega;\hbar,\tau)$ without a Hamiltonian, equilibrium, and collision mechanism.

\subsection{The Hamiltonian limit and the Kubo expression}
For a finite sample in a length gauge, let $\hat P_j$ be the total electric dipole and $\delta H=-\hat P_jE_j$. The Hamiltonian source is $\ii[\hat P_j,\rho_0]/\hbar$. In a closed system, the response kernel is the retarded commutator of the measured current and the perturbing dipole \cite{kubo1957}.

The resolvent formula reduces to
\begin{equation}\label{eq:length_kubo}
 \sigma_{ij}(s)=\frac{\ii}{\hbar\mathcal V}\int_0^\infty e^{-st}
 \langle[\hat I_i(t),\hat P_j(0)]\rangle_0\,\dd t,
\end{equation}
where $\hat I_i=\dot{\hat P}_i$ is an extensive current. A periodic infinite conductor requires a twist, a vector potential, or another well-defined thermodynamic-limit construction; the unbounded position operator cannot be used naively as a Bloch-space dipole.

In an energy eigenbasis, any response entry has a Lehmann-type decomposition with denominators $s+\ii(E_a-E_b)/\hbar$ and occupation differences. If $\cL_0$ is diagonal on those coherence operators, the denominators become $s+\gamma_{ab}+\ii(E_a-E_b)/\hbar$. If it is not diagonal, the matrix resolvent must be retained. Replacing every denominator by the same scalar decay rate is an additional approximation, not a consequence of quantum mechanics.

\section{Projection and the validity of the three-state closure}\label{app:projection}
A useful positive equilibrium inner product on observables is the Kubo--Mori form
\begin{equation}\label{eq:KM}
 (A,B)_0=\int_0^1\Tr\left(\rho_0^{1-u}\delta A^\dagger\rho_0^u\delta B\right)\dd u,
 \qquad\delta A=A-\Tr(\rho_0A).
\end{equation}
For independent Hermitian observables $X_a$, their covariance matrix $G_{ab}=(X_a,X_b)_0$ is positive on the chosen subspace. Conserved or redundant variables must be treated before inverting it. This metric has a physical equilibrium origin; a Euclidean norm on unscaled current and coherence variables would compare quantities with incompatible units.

Projection and memory-function methods for transport are established \cite{mori1965}. We use their algebraic block-resolvent content without assuming that all omitted variables decay rapidly.

Let $\cP$ be a projection onto selected linear perturbations and $\cQ=1-\cP$. Writing $R(s)=(s-\cL_0)^{-1}$, block elimination gives
\begin{align}
 \cP R(s)\cP
 &=\left[s-\cL_{PP}-\Sigma_Q(s)\right]^{-1}_{P},\label{eq:Feshbach}\\
 \Sigma_Q(s)&=\cL_{PQ}(s-\cL_{QQ})^{-1}\cL_{QP}.\label{eq:Sigma}
\end{align}
When the drive has an omitted component, the effective source is
\begin{equation}\label{eq:effective_source}
 S_P^{\mathrm{eff}}(s)=S_P+\cL_{PQ}(s-\cL_{QQ})^{-1}S_Q.
\end{equation}
There are corresponding output corrections if the measured current has a $Q$ component. Retaining only \eqref{eq:Sigma} and forgetting those vertices need not reproduce the exact conductivity.

\subsection{A quantitative condition for a closure}
\paragraph{Closure condition.}\label{ass:Q}
On a specified compact response region $\mathcal K$ in the right half-plane, the omitted resolvent is bounded. A proposed reduced inverse response $Z_0(s)$ is invertible there. The residual inverse-response error is $\Delta Z(s)$, including the difference between the exact and approximate memory.

\begin{proposition}[Resolvent error bound]\label{prop:error}
If $\eta=\sup_{s\in\mathcal K}\|Z_0(s)^{-1}\Delta Z(s)\|<1$, then
\begin{equation}\label{eq:resolvent_error}
 \|[Z_0+\Delta Z]^{-1}-Z_0^{-1}\|
 \leq\frac{\|Z_0^{-1}\|^2\|\Delta Z\|}{1-\eta}
\end{equation}
pointwise on $\mathcal K$. Source and output errors add the corresponding products with these bounded resolvents.
\end{proposition}
\begin{proof}
Factor $Z_0+\Delta Z=Z_0(I+Z_0^{-1}\Delta Z)$ and apply the convergent Neumann series. The resolvent identity gives the stated difference bound.
\end{proof}
For example, if $\|e^{t\cL_{QQ}}\|\leq C_Qe^{-\Lambda_Qt}$, then
\begin{equation}
 \|(s-\cL_{QQ})^{-1}\|\leq\frac{C_Q}{\Rea s+\Lambda_Q},\qquad
 \|\Sigma_Q(s)\|\leq\frac{C_Q\|\cL_{PQ}\|\|\cL_{QP}\|}{\Rea s+\Lambda_Q}.
\end{equation}
A spectral gap alone is insufficient for this estimate in a nonnormal sector; a semigroup or resolvent bound is needed. This bound concerns the specified response window; it does not establish accurate continuation to ill-conditioned poles.

\subsection{Microscopic quantities required to identify the reduced model}
To derive the closure from a material model, one must compute the susceptibility matrix, the projected Hamiltonian commutators, the projected bath action, and the omitted-sector memory. A coherence frequency is an energy gap divided by $\hbar$, but an effective coupling need not be a bare matrix element divided by $\hbar$ after metric normalization. In particular, symmetry can force that coupling to vanish.

In a normalized invariant subspace with basis perturbations $R_a$ and dual observables $X_a$ satisfying $\Tr(X_aR_b)=\delta_{ab}$, the matrix is
\begin{equation}\label{eq:projected_entries}
 A_{ab}=\Tr(X_a\cL_0R_b),\quad
 b_a=\Tr(X_aS_E),\quad c_a=\Tr(\hat J R_a).
\end{equation}
The three-state closure is exact for that microscopic model only if the subspace is invariant, omitted drive and output terms vanish or are accounted for, and these entries have the stated structure. Otherwise it is an approximation whose error must be assessed using the full response. Positive-real macroscopic dynamics do not alone prove that a specific GKSL generator realizes them.

An operator construction can test the electronic interpretation before fitting. In the equilibrium metric, form the normalized current $\hat\jmath_0$, its orthogonal reversible force $F=\mathcal Q_j\cL_H\hat\jmath_0$, and a second orthogonal force $R=\mathcal Q_{ju}\cL_H(F/\|F\|_0)$. Their norms define candidate coupling scales, while subsequent forces generate omitted memory. With band projectors $P_n$, the fraction $\|\sum_{n\ne m}P_nFP_m\|_0^2/\|F\|_0^2$ measures the interband content of the current force in that chosen decomposition. It does not by itself prove that the remaining hierarchy closes. Conserved sectors must be projected out consistently.

\section{A reservoir realization and its limited interpretation}\label{app:bath}
\subsection{Hamiltonian transport-mode model}
The reduced response has a quantum Hamiltonian realization through a structured harmonic reservoir. Oscillator reservoirs and quantum Langevin equations are established constructions \cite{caldeira1983,ford1988}. The spectral density below is chosen to produce the memory in \eqref{eq:M}. This is an existence construction for a quantum transport coordinate, rather than a derivation from a particular electronic band.

Let $[X,P]=\ii\hbar$, with transport-coordinate mass $M$ and charge $Q$. The bulk current is $J=Q\dot X/\mathcal V$, and $D=Q^2/(M\mathcal V)$. A collective drift coordinate with $Q=-Ne$ and $M=Nm$ gives $D=ne^2/m$, although this representation does not specify the remaining fermionic degrees of freedom. For neutral reservoir oscillators $[q_a,p_b]=\ii\hbar\delta_{ab}$, use a uniform length gauge:
\begin{equation}\label{eq:bath_H}
 H_{\mathrm{tot}}=\frac{P^2}{2M}+\frac{M\omega_0^2X^2}{2}-QXE(t)
 +\sum_a\left[\frac{p_a^2}{2m_a}
 +\frac{m_a\omega_a^2}{2}\left(q_a-\frac{c_aX}{m_a\omega_a^2}\right)^2\right].
\end{equation}
The square supplies the counterterm cancelling an unintended bath-induced static restoring force. A positive $\omega_0$ makes the finite-bath equilibrium normalizable; the free-transport response uses $\omega_0\to0$ at fixed nonzero $s$. Velocity-gauge coupling uses $P\mapsto P-QA$ with its contact current. This uniform channel model does not derive the finite-wave-vector compressibility closure.

\subsection{Exact reservoir elimination}

The Heisenberg solution of each reservoir oscillator is
\begin{align}
 q_a(t)&=q_a(0)\cos\omega_at+\frac{p_a(0)}{m_a\omega_a}\sin\omega_at\nonumber\\
 &\quad+\frac{c_a}{m_a\omega_a}\int_0^t\sin[\omega_a(t-t')]X(t')\dd t'.
\end{align}
Integration by parts gives the exact quantum Langevin equation
\begin{equation}\label{eq:QLangevin}
 M\ddot X+M\omega_0^2X+\int_0^t\mu(t-t')\dot X(t')\dd t'
 =QE(t)+F(t)-\mu(t)X(0),
\end{equation}
with
\begin{align}
 \mu(t)&=\sum_a\frac{c_a^2}{m_a\omega_a^2}\cos\omega_at,\\
 F(t)&=\sum_ac_a\left[q_a(0)\cos\omega_at+
 \frac{p_a(0)}{m_a\omega_a}\sin\omega_at\right].
\end{align}
The last term is the initial slip. For a zero-mean preparation with $\langle X(0)\rangle=\langle F(t)\rangle=0$, it does not affect the mean driven response. For arbitrary preparation it must be retained. Thermal stationary response can instead be formulated from the remote past, with a correlated equilibrium preparation.

Define the positive reservoir spectral density
\begin{equation}\label{eq:bath_spectral}
 \mathcal I(\nu)=\frac\pi2\sum_a\frac{c_a^2}{m_a\omega_a}\delta(\nu-\omega_a),\quad \nu>0.
\end{equation}
The friction transform is
\begin{equation}\label{eq:bath_transform}
 \frac{\widetilde\mu(s)}M=\frac{2}{\pi M}\int_0^\infty
 \frac{\mathcal I(\nu)}\nu\frac{s}{s^2+\nu^2}\dd\nu.
\end{equation}
The elimination is exact for the Hamiltonian. It does not assume a time-local GKSL equation for the reduced coordinate state.

\subsection{A positive spectral density for the proposed memory}

Choose
\begin{equation}\label{eq:structured_spectral}
 \frac{\mathcal I_c(\nu)}{M\nu}=\frac{g^2\gc}{2}\left[
 \frac1{\gc^2+(\nu-\Omega)^2}+\frac1{\gc^2+(\nu+\Omega)^2}\right]\geq0.
\end{equation}
The cosine transform of this even Lorentzian pair gives
\begin{equation}
 \mu_c(t)/M=g^2e^{-\gc t}\cos\Omega t,\qquad
 \widetilde\mu_c(s)/M=M_c(s).
\end{equation}
An ideal Ohmic part $\mathcal I_m(\nu)=M\gm\nu$ gives $\widetilde\mu_m/M=\gm$. A finite Drude cutoff instead gives
\begin{equation}\label{eq:cutoff}
 \mathcal I_{m,\Lambda}(\nu)=M\gm\nu\frac{\Lambda^2}{\nu^2+\Lambda^2},\qquad
 \widetilde\mu_{m,\Lambda}(s)/M=\frac{\gm\Lambda}{s+\Lambda}.
\end{equation}
The regulated quantum mean response is therefore
\begin{equation}\label{eq:regulated_sigma}
 \sigma_{\Lambda,\omega_0}(s)=
 \frac{D}{s+\omega_0^2/s+\gm\Lambda/(s+\Lambda)+M_c(s)}.
\end{equation}
It converges to \eqref{eq:sigma} as $\omega_0\to0$, $\Lambda\to\infty$ at fixed $s$ in the right half-plane. The inverse-response corrections are $\omega_0^2/s$ and $-\gm s/(s+\Lambda)$; both are controlled on compact response windows away from the static singularity and cutoff pole.

An ideal Ohmic spectrum has an ultraviolet-divergent bare counterterm, requiring a regulated limit. Broadband fluctuations and arbitrarily short times must retain the cutoff or specify renormalization. A finite cutoff changes the reduced response outside the stated limiting window.

The parameter $\gc$ in this construction is a reservoir correlation linewidth, not a derived electronic pure-dephasing rate. A thermal quantum reservoir also fixes quantum noise, whereas the same deterministic mean response can be obtained from classical coordinates. The existence of this realization therefore establishes compatibility with a quantum Hamiltonian, not uniqueness of an electronic-coherence mechanism.

\section{Passivity and electromagnetic determinant checks}\label{app:det}
For $A=K-\Gamma$, $K^\dagger=-K$, $\Gamma=\diag(\gm,\gc,\gc)>0$, and $x=(sI-A)^{-1}b$,
\begin{equation}
 \Rea[b^\dagger(sI-A)^{-1}b]
 =(\Rea s)\|x\|^2+x^\dagger\Gamma x\geq0.
\end{equation}
Every material eigenvalue satisfies $\Rea s_n\leq-\min(\gm,\gc)$. Retardation and analyticity are properties of this stable response; no separate relativistic-causality claim follows.
The characteristic matrix is
\begin{equation}
 sI-A=\begin{pmatrix}s+\gm&g&0\\-g&s+\gc&\Omega\\0&-\Omega&s+\gc\end{pmatrix}.
\end{equation}
Expansion by the first row gives $P=(s+\gm)[(s+\gc)^2+\Omega^2]+g^2(s+\gc)$. The $(1,1)$ cofactor is $N$, hence $(sI-A)^{-1}_{11}=N/P$ without any elimination assumption about nonzero $s+\gc$. This cofactor formula remains valid by analytic continuation wherever the resolvent exists, including where an intermediate step dividing by $s+\gc$ would be inappropriate.

For a longitudinal mode with Gauss's law imposed, the state matrix at $k>0$ is
\begin{equation}\label{eq:AL}
 A_L(k)=\begin{pmatrix}
 0&-\ii c_hk&0&0\\
 -\ii\left(c_hk+\frac{\omega_p^2}{c_hk}\right)&-\gm&-g&0\\
 0&g&-\gc&-\Omega\\
 0&0&\Omega&-\gc
 \end{pmatrix}.
\end{equation}
Its non-Euclidean density metric is the physical one inherited from the electrostatic field energy:
\begin{equation}
 H_L=\diag\left(1+\frac{\omega_p^2}{c_h^2k^2},1,1,1\right).
\end{equation}
Indeed $A_L^\dagger H_L+H_LA_L=-2\diag(0,\gm,\gc,\gc)$. Thus the $1/k$ entry created by Gauss elimination is not an instability; it reflects the nonlocal electrostatic energy. Direct expansion gives $\det(sI-A_L)=\mathcal P_L$.

For one transverse polarization use normalized fields $e_T=\sqrt{\epsb}E_T$ and $h_T=\sqrt{\mu_b}H_T$. A consistent orientation gives
\begin{equation}\label{eq:AT}
 A_T(k)=\begin{pmatrix}
 0&-\ii v_bk&-\omega_p&0&0\\
 -\ii v_bk&0&0&0&0\\
 \omega_p&0&-\gm&-g&0\\
 0&0&g&-\gc&-\Omega\\
 0&0&0&\Omega&-\gc
 \end{pmatrix}.
\end{equation}
Its Hermitian part is $-\diag(0,0,\gm,\gc,\gc)$, and $\det(sI-A_T)=\mathcal P_T$. A change of polarization orientation changes signs in paired off-diagonal entries but not the determinant or the energy estimate.

\section{Optional repeated-pole limit}\label{app:EP}
The material matrix can have a third-order exceptional point, but this special tuning is not needed for any conductivity or Maxwell result in the main text. For $\delta=\gm-\gc\ne0$, coefficient comparison in $P(s)$ gives
\begin{equation}\label{eq:EP3}
 \Omega^2=\delta^2/27,\qquad g^2=8\delta^2/27,\qquad
 s_*=-(\gm+2\gc)/3.
\end{equation}
Both couplings are nonzero. The first two rows of the eigenvalue equation then determine the complete eigenvector from its current component, so the eigenspace is one-dimensional and the triple root has a size-three Jordan block. Its current impulse response is
\begin{equation}
 h_*(t)=D e^{s_*t}\left(1-2\delta t/3+2\delta^2t^2/27\right),\qquad t\geq0.
\end{equation}
This is a passive finite-state property, not a uniquely quantum signature. Moreover $\mathcal P_T(k,s_*)=\omega_p^2s_*N(s_*)\ne0$ and $\mathcal P_L(k,s_*)=(c_h^2k^2+\omega_p^2)N(s_*)\ne0$: a material exceptional point is not an exceptional point of the coupled Maxwell system. Within the conditional orbital map it additionally requires $g^2=8\Omega_0^2/9$ and $|\delta|=\sqrt3\Omega_0$, so independent tuning of the phenomenological coefficients is not justified. No realization of this tuning in a material is established here.


\end{document}